\documentclass{rQUF2e}
\usepackage{amssymb,amsmath,amsthm}
\usepackage{setspace}
\usepackage{tikz}
\usepackage{bbm}
\usepackage[latin1]{inputenc}
\usepackage{eurosym} 
\usepackage{xcolor}
\usepackage{tabularx} 
\usepackage{caption}
\usepackage{rotating}

\newcommand\ytl[2]{
\parbox[b]{8em}{\hfill{\color{black}\bfseries\sffamily #1}~$\phantom{\cdots\cdots}$~}\makebox[0pt][c]{$\bullet$}\vrule\quad \parbox[c]{9.0cm}{\vspace{7pt}\color{red!0!black!100}\raggedright\sffamily #2.\\[7pt]}\\[-3pt]}

\newtheorem{thm}{Theorem}
\newtheorem{lem}{Lemma}
\newtheorem*{remark}{Remark} 

\newcommand\blfootnote[1]{%
  \begingroup
  \renewcommand\thefootnote{}\footnote{#1}%
  \addtocounter{footnote}{-1}%
  \endgroup
}

\AtEndDocument{\newpage\bigskip{\footnotesize%
\noindent   $^1$ \textsc{Department of Mathematics and Statistics, University of Helsinki, FI-00014 Helsinki, Finland}. \par
\noindent \textit{E-mail address}: \texttt{petteri.piiroinen@helsinki.fi}\par
  \addvspace{\medskipamount}
  \noindent $^2$ \textsc{School of Engineering Science, Lappeenranta-Lahti University of Technology, FI-53850 Lappeenranta, Finland}.  \par
	 \noindent \textit{E-mail address}: \texttt{lassi.roininen@lut.fi} \par
  \addvspace{\medskipamount}
  \noindent   $^3$ \textsc{Deka Investment GmbH, 60325 Frankfurt am Main, Germany}. \par 
 \noindent M.~Simon \textit{E-mail address}: \texttt{info@simon-martin.net}
  
}}
\begin{document}

\title{Brexit Risk Implied by the SABR Martingale Defect in the EURGBP Smile}
\author{Petteri Piiroinen$^1$, Lassi Roininen$^2$ and Martin Simon$^3$ \\
\affil{$^1$University of Helsinki, Finland \\ $^2$Lappeenranta-Lahti University of Technology, Finland \\ $^3$Deka Investment GmbH, Germany}}
\maketitle

\blfootnote{The opinions expressed in this article are those of the authors and do not necessarily reflect views of Deka Investment GmbH.}

\begin{abstract}
We construct a data-driven statistical indicator for quantifying the tail risk perceived by the EURGBP option market surrounding Brexit-related events. We show that under lognormal SABR dynamics this tail risk is closely related to the so-called martingale defect and provide a closed-form expression for this defect which can be computed by solving an inverse calibration problem. 
In order to cope with the the uncertainty which is inherent to this inverse problem, we adopt a Bayesian statistical parameter estimation perspective. We probe the resulting posterior densities with a combination of optimization and adaptive Markov chain Monte Carlo methods, thus providing a careful uncertainty estimation for all of the underlying parameters and the martingale defect indicator.
Finally, to support the feasibility of the proposed method, we provide a Brexit \lq\lq fever curve\rq\rq\ for the year 2019.

\end{abstract}


\begin{classcode} C02, D81, G10, G13\end{classcode}

\section{Introduction}
The British exit (Brexit) from the European Union and the uncertainty
surrounding its modalities have had an impact on financial markets during the
last years. Data-driven approaches to infer market expectations with regard to
Brexit-related events which use option market data have been proposed recently
by Clark and Amen \citeyearpar{Clark} and by Hanke, Poulsen and Weissensteiner
\citeyearpar{Hanke}.

In this work, we propose a novel forward-looking indicator based on foreign
exchange (FX) option market data for quantifying the market expectations with
regard to tail risks around important Brexit-related events. The tail risk
observed in FX option markets, i.e., the risk of outlier returns two or more
standard deviations below the mean, is significantly greater than the risk
obtained under the theoretical assumption that the returns of the underlying
currency pair follows a lognormal distribution. Our so-called \emph{SABR GBP
martingale defect indicator} has first been introduced in the context of
quantifying the risk of stock price bubbles, cf. \cite{Piiroinen}. It is
derived from the market price of FX tail risk related to a devaluation of the
GBP against the EUR which is in turn calculated from the prices of
out-of-the-money EURGBP options with maturity close to the event of interest.
In 2019 this indicator typically ranged from $0\%$ to $20\%$ and as it rises,
the corresponding tail of the return distribution acquires more weight such
that the probabilities of extreme outlier returns become more significant.
Within the lognormal SABR model, at some point the underlying process becomes a
strict local martingale, see Theorem \ref{thm:1} below, and it as been shown by
Jacquier and Keller-Ressel in \cite{Jacquier} that this is indeed a
model-independent result: An increase in perceived tail risk increases the
relative demand for out-of-the-money options leading to a steepening of the
slope of the implied volatility smile. One of the key results in
\cite{Jacquier} is that, under the assumption of fully collateralized trades,
the total implied variance for a fixed time to maturity in log-strike space
attains an asymptotic slope of $2$ if and only if the discounted underlying
stochastic process is a strict local martingale. Strict local martingales have
in financial applications been mainly employed to model stock price bubbles in
financial markets, see, e.g., \citep{CoxHobson, Jarrow, Jarrow3, Jarrow4,
Jarrow5, Jarrow6, Kardaras, Piiroinen, Protter1}. For applications the context
of FX modeling see, e.g., \citep{Carr,Chaim}.
We introduce here a statistical martingale defect indicator based on bid and
ask implied volatility quotes which is carefully tailored to account for the
uncertainty inherent in FX option markets stemming from the availability of
merely 5 implied volatility quotes available for each time to maturity. To make
things even more uncertain, the 10-delta quotes which are particularly
important with regard to the tail behavior presumably have only limited
reliability because they are derived from less liquid options; for a discussion
of the topic we refer the reader to \cite{Wystup1}. 
Our indicator provides a simple and reliable tool to quantify perceived tail risk while at the same time accounting for these inherent uncertainties which is important when it comes to data-driven generation of relevant risk and stress test scenarios for risk management purposes. On top of that we would like to stress the fact that using the SABR model to valuate FX options is a widespread practice in the market so that integration of our indicator into existing risk management infrastructure should be straight forward. 

The outline of this paper is as follows. In the following section the notation and the mathematical setting will be introduced. In Section 3, we study the presence of strict local martingales in the lognormal SABR model for the underlying currency pair EURGBP and derive the corresponding martingale defect indicator. Section 4 is devoted to the introduction of the statistical framework for uncertainty quantification and in Section 5 we present a Brexit \lq\lq fever curve\rq\rq\ along the timeline of relevant events in 2019 which is enclosed in Appendix A. We conclude with a discussion of our findings in Section 6.

\section{Notation and mathematical setting}
Let $(\Omega,\mathcal{F},\mathcal{F}_t,\mathbbm{P}_i)$, $i=\text{\pounds},\text{\euro{}},$ denote a filtered probability space such that $\mathcal{F}_t$ satisfies the usual assumptions. On this probability space we will define the stochastic process $\{S_t,t\geq 0\}$ to model the foreign exchange rate in the usual FOR-DOM convention capturing a DOM investor's perspective and the process $\{\widehat{S}_t:=S_t^{-1},t\geq 0\}$ for the DOM-FOR exchange rate which corresponds to the point of view of a FOR investor. That is, $S_t$ denotes the number of units of domestic currency (DOM) required to buy one unit of foreign currency (FOR) at time $t$ and vice versa for $\widehat{S}_t$. In this work we are interested in the currency pair EURGBP, where EUR is the foreign currency and GBP the domestic one. We assume that for each currency a risk-free money market account exists such that 
\begin{equation*}
dB_{i}(t) = r_{i}(t)\mathrm{d}t,\quad B_{i}(0)=1,\ i=\text{\pounds},\text{\euro{}},
\end{equation*}
where $r_{\text{\euro{}}}$ denotes the time-dependent continuously compounded foreign interest rate and $r_{\text{\pounds}}$ denotes the time-dependent continuously compounded domestic interest rate (for the sake of readability we will suppress the time-dependence in our notation). The DOM investor can trade in the domestic money market account $B_{\text{\pounds}}(t)$ or the foreign money market account $B_{\text{\euro{}}}(t)S_t$, whereas the FOR investor may trade in the foreign money market account $B_{\text{\euro{}}}(t)$ or the domestic money market account $B_{\text{\pounds}}(t)\widehat{S}_t$. We denote by $\mathbbm{P}_{\text{\pounds}}$ a \emph{domestic equivalent martingale measure}, i.e., a probability measure such that 
\begin{equation}\label{eqn:elmm}
\mathbbm{E}^{\mathbbm{P}_{\text{\pounds}}}\left\{\frac{B_{\text{\euro{}}}(T)S_T}{B_{\text{\pounds}}(T)}\Big|\mathcal{F}_t\right\}\leq \frac{B_{\text{\euro{}}}(t)S_t}{B_{\text{\pounds}}(t)}
\end{equation} 
and by $\mathbbm{P}_{\text{\euro{}}}$ a \emph{foreign equivalent martingale measure}, i.e., a probability measure such that 
\begin{equation}
\mathbbm{E}^{\mathbbm{P}_{\text{\euro{}}}}\left\{\frac{B_{\text{\pounds}}(T)S_T^{-1}}{B_{\text{\euro{}}}(T)}\Big|\mathcal{F}_t\right\}\leq\frac{B_{\text{\pounds}}(t)S_t^{-1}}{B_{\text{\euro{}}}(t)}.
\end{equation} 
For the rest of this work we assume that a domestic equivalent martingale measure $\mathbbm{P}_{\text{\pounds}}$ which satisfies (\ref{eqn:elmm}) exists or in other words that the process 
\begin{equation*}
\left\{\frac{B_{\text{\euro{}}}(t)S_t}{S_0B_{\text{\pounds}}(t)},t\geq 0\right\}
\end{equation*}
is a local $\mathbbm{P}_{\text{\pounds}}$-martingale. This implies that the
market model satisfies NFLVR, cf. Delbean and Schachermayer \cite{Delbean1}. We
do not in general assume that $\mathbbm{P}_{\text{\pounds}}$ is unique as we
are going to work with the SABR model which is an incomplete market model -- in
this setting the market can be completed in the sense that calibration to
observed option market data chooses a particular equivalent martingale measure.

An outright forward contract trades at time $t$ at zero cost and leads to an exchange of notional at time $T$ at the pre-specified outright forward rate
\begin{equation*}
F_t(T) = S_t\cdot e^{(r_{\text{\pounds}}-r_{\text{\euro{}}})(T-t)}.
\end{equation*}
At time $T$, the foreign notional amount $N$ will be exchanged against an amount of $NF_t(T)$ domestic currency units. 
FX options are usually physically settled, that is the buyer of a EUR European plain vanilla call with strike $K$ and time to maturity $T$ receives a EUR notional amount $N$ and pays $NK$ GBP. The value of such plain vanilla options is computed via the standard Black-Scholes-Merton formula
\begin{equation*}
V_t(K,T,\phi)=\textrm{BSM}(F_t;K,T,\phi)=\phi e^{-r_{\text{\pounds}}}(T-t)\{F_t(T)\mathcal{N}(\phi d_+)-K\mathcal{N}(\phi d_-)\},
\end{equation*}
where $\mathcal{N}(\cdot)$ denotes the cumulative normal distribution function, $\phi=\pm 1$ for a call, respectively put option and 
\begin{equation*}
d_{\pm} =\frac{\log(\frac{F_t(T)}{K})\pm\frac{1}{2}\sigma^2(T-t)}{\sigma\sqrt{T-t}}
\end{equation*}
with the Black-Scholes-Merton volatility $\sigma$. This volatility for every strike and time to maturity can be implied from plain vanilla option prices 
\begin{equation*}
\sigma^{\textrm{M}}(F_t;K,T,\phi)=\textrm{BSM}^{-1}(V_t^{\textrm{M}}(K,T,\phi),F_t;K,T,\phi).
\end{equation*}
In FX markets, vanilla option prices are commonly quoted via an at-the-money
straddle volatility together with quotes for $10$-delta and $25$-delta risk
reversals respectively strangles with expiry dates corresponding to overnight
maturity, $1$, $2$ and $3$ weeks and $1$, $2$, $3$, $4$, $5$, $6$, $9$ and $12$
months. Quoting conventions vary depending on the underlying currency pair,
expiry and broker, see, e.g., \cite{Reiswich1,Reiswich2,Wystup2}. We use in
this work preprocessed market data obtained from Refinitiv Financial Solutions
which provides composite implied volatilities derived from different
contributing broker sources versus the corresponding deltas. We converted this
data into strike space in line with the used market conventions with regard to
deltas and ATM definition, see \cite{Reiswich1,Wystup2} for a detailed
description of delta-strike conversion.
To sum up, we have for each time to maturity $T$ a vector $\boldsymbol{y}$ which consists of five mid implied volatility data points 
\begin{equation*}
\boldsymbol{y}:=\left[\begin{array}{l}
\sigma^{\textrm{M}}(F_t;K_{10,-1},T)\\
\sigma^{\textrm{M}}(F_t;K_{25,-1},T)\\
\sigma^{\textrm{M}}(F_t;K_{\text{ATM}},T)\\
\sigma^{\textrm{M}}(F_t;K_{25,1},T)\\
\sigma^{\textrm{M}}(F_t;K_{10,1},T)\end{array}\right]
\end{equation*}
with strikes $K_{x,\phi}$, $x\in\{10,25\}$, chosen such that the corresponding call (for $\phi=1$) and put (for $\phi=-1$) option has a delta of $x$, together with the corresponding vector of bid ask spreads ${\textbf{BA}}\in\mathbbm{R}^5$.

\section{The SABR Martingale Defect in FX Smiles}
We fit the volatility smile using the stochastic alpha, beta, rho, or brief SABR model. That is, we assume the following dynamics for the forward process under the domestic equivalent martingale measure $\mathbbm{P}_{\text{\pounds}}$:
\begin{eqnarray}\label{eqn:SABR1}
\mathrm{d}F_t(T)&=&\alpha_t F_t(T)^{\beta}\mathrm{d}W_t^{(1)}\\
\label{eqn:SABR2}\mathrm{d}\alpha_t&=&\nu\alpha_t \mathrm{d}W_t^{(2)},
\end{eqnarray}
with fixed elasticity parameter $\beta=1$ (which is a common choice for FX smile modeling), volatility of volatility $\nu>0$ and two correlated $\mathbbm{P}_{\text{\pounds}}$-Brownian motions $W^{(1)}$ and $W^{(2)}$ with correlation parameter $\rho\in[-1,1]$. 
The dynamics of $\widehat{F}_t(T):=F_t(T)^{-1}$ under the foreign equivalent martingale measure $\mathbbm{P}_{\text{\euro{}}}$ are given by 
\begin{eqnarray}\label{eqn:SABR3}
\mathrm{d}\widehat{F}_t(T)&=&\alpha_t \widehat{F}_t(T)^{2-\beta}\mathrm{d}\widehat{W}_t^{(1)}\\
\label{eqn:SABR4}\mathrm{d}\alpha_t&=&\rho\nu\alpha_t^2\widehat{F}_t(T)^{1-\beta} \mathrm{d}t+\nu\alpha_t \mathrm{d}\widehat{W}_t^{(2)},
\end{eqnarray}
with two correlated $\mathbbm{P}_{\text{\euro{}}}$-Brownian motions $\widehat{W}^{(1)}$ and $\widehat{W}^{(2)}$ with correlation parameter $-\rho$. 

\vspace{.2cm}
The following Theorem provides the theoretical foundation of our approach.
\begin{thm}\label{thm:1}
Assume that the forward process $\{F_t(T), 0\leq t\leq T\}$ follows the SABR dynamics (\ref{eqn:SABR1}), (\ref{eqn:SABR2}), then we have for all $t\in[0,T]$: 
\begin{equation}\label{thm:1:eq1}
F_0(T)\mathbbm{P}_{\text{\euro{}}}\{\widehat{F}_t(T)>0\}=\mathbbm{E}^{\mathbbm{P}_{\text{\pounds}}}\{F_t(T)|F_0(T)\}.
\end{equation}
Moreover, we have for all $t\in(0,T]$: 
\begin{equation}
\mathbbm{E}^{\mathbbm{P}_{\text{\pounds}}}\{F_t(T)|F_0(T)\}<F_0(T)\text{
  if and only if }\rho>0. 
\end{equation}
\end{thm}
\pagebreak
\begin{proof}
  The forward $F$ and stochastic volatility $\alpha$ given by
  equations~\eqref{eqn:SABR1} and~\eqref{eqn:SABR2} are both positive exponential
  semimartingales
  \begin{equation}
    F_t = F_0 \mathcal E(\alpha \cdot W^{(1)})_t \quad\text{and}\quad
    \alpha_t = \alpha \mathcal E(\nu W^{(2)})_t,
  \end{equation}
  respectively. Here we denote the exponential semimartingale of $X$
  with $X_0 = 0$ as
  \[
    \mathcal E(X)_t = \exp\left(X_t - \frac12 \langle X \rangle_t\right)
  \]
  and the $H \cdot X$ denotes the stochastic integral with respect to a
  semimartingale 
  \[
    (H \cdot X)_t = \int_0^t H(s) \mathrm{d} X_s.
  \]
  By Fatou's Theorem, the positive exponential semimartingale $\mathcal
  E(X)$ is a martingale if and only if $\mathbbm{E}^{\mathbbm{P}_{\text{\pounds{}}}}\{\mathcal E(X)_t\} = 1$ for
  every $t > 0$. 
  As Cox and Hobson~\cite{CoxHobson} point out, the martingale property
  follows with an argument due to Sin~\cite{Sin98}.
  According to this work, when $\beta = 1$, the
  expectation of the exponential semimartingale $F_t/F_0 = \mathcal
  E(\alpha \cdot W^{(1)})$ on the interval $[0, T]$ is given by
  \[
    \mathbbm{E}^{\mathbbm{P}_{\text{\pounds{}}}}\{ \mathcal E(\alpha \cdot W^{(1)})_t\} = \mathbbm{P}\{
    \widehat{\tau}_\infty > t\}
  \]
  for every $t > 0$ where under the $\mathbbm{P}_{\text{\euro}}$-probability the stopping time
  $\widehat{\tau}_\infty$ is the time of explosion of the auxiliary
  process that under $\mathbbm{P}_{\text{\euro}}$-probability satisfies the following
  SDE
  \[
    \mathrm{d} v_t = \nu v_t \mathrm{d} W_t^{(3)} + \nu \rho v_t^2 \mathrm{d} t,\quad v_0 =
    \alpha,
  \]
  where $W^{(3)}$ is a standard Brownian motion under $\mathbbm{P}_{\text{\euro}}$.  We
  have to open this up bit more, since the argument reveals the stated
  properties.

  \noindent First of all, assuming (for the sake of simplicity) the stochastic volatility is
  bounded, then the Novikov condition implies that the forward $F_t$
  is a uniformly integrable martingale on $[0, T]$.
  In this case, the Girsanov Theorem implies that the foreign equivalent
  martingale measure $\mathbbm{P}_{\text{\euro{}}}$ is given by
  \[
    \mathbbm{P}_{\text{\euro{}}} \{A\} = 
    \mathbbm{E}^{\mathbbm{P}_{\text{\pounds{}}}} \{F_T [ A ]\}
  \]
  and the boundedness of stochastic volatility implies that both $F$ and
  $\widehat{F}$ stay bounded on $[0,T]$, i.e., neither the domestic nor
  the foreign forward reach zero before time $T$ with probability 1
  under either probability.

  Therefore, in order that either one is a strict local martingale the
  stochastic volatility should attain unbounded values. The analysis of
  the SDE already implies that this is the case and moreover, there are
  no explosions on the bounded interval $[0, T]$.
  Defining the stopping times 
  \[
    \tau_n = \inf\{\, t \in (0, T] \, ; \, \alpha_t \ge n \,\},\quad n\in\mathbbm{N}
  \]
  and stopped foreign equivalent martingale measures
  \[
    \mathbbm{P}_{n,\text{\euro{}}} \{A\} = 
    \mathbbm{E}^{\mathbbm{P}_{\text{\pounds{}}}} \{F_T^{\tau_n} [ A ]\}
  \]
  we obtain as in \cite{Sin98} with Dominated Convergence Theorem and Girsanov
Theorem 
  \[
    \mathbbm{E}^{\mathbbm{P}_{\text{\pounds{}}}} \{F_t\}
    =
    \lim_{n \to \infty} \mathbbm{P}_{n,\text{\euro{}}}\{\tau_n > t\}.
  \]
Note that the stochastic volatility under the stopped foreign equivalent martingale
  measures $\mathbbm{P}_{n,\text{\euro{}}},$ $n\in\mathbbm{N}$ satisfies the SDE
  \[
    d\alpha_t = [\, t \le \tau_n \,] \rho \nu \alpha_t^2 dt
	      + [\, t \le \tau_n \,] \nu\alpha_t dW_t^{(2)}.
  \]
  Note moreover, that if $\tau_n > t$, then both $F_t > 0$ and $\widehat{F}_t > 0$,
  so 
  \[
    \liminf_{n \to \infty} \mathbbm{P}_{n,\text{\euro{}}}\{\widehat F_t > 0\}
    \ge \liminf_{n \to \infty} \mathbbm{P}_{n,\text{\euro{}}}\{\tau_n > t\}.
  \]
  This estimate implies that if the limit $\mathbbm{P}_{n,\text{\euro{}}}\{\tau_n
  > t\}$ is one, then the forward process $F_t$ is a
  $\mathbbm{P}_{\text{\pounds{}}}$-martingale and 
  $\mathbbm{P}_{\text{\euro{}}}\{\widehat F_t > 0\} = 1$.

  Since there are clearly no explosions, when $\rho = 0$ and by Comparison
  Theorem, when $\rho \le 0$, we can estimate 
  \[
    \mathbbm{P}_{n,\text{\euro{}}}\{\tau_n^{(\rho)} > t\}
    \ge \mathbbm{P}_{n,\text{\euro{}}}\{\tau_n^{(0)} > t\}
  \]
  so there are no explosions when $\rho < 0$ either. We have 
  shown in~\cite{Piiroinen} that when $\rho > 0$ the explosion does occur,
namely
  \[
    \mathbbm{E}^{\mathbbm{P}_{\text{\pounds{}}}} \{F_t\}
    = \mathbbm{P}_{\text{\euro{}}} \{\tau_\infty > t\} < 1.
  \]
  We already deduced that if there are no explosions on the interval
  $[0, T]$, then necessarily $\widehat F_t > 0$, so 
  \[
    \mathbbm{P}_{\text{\euro{}}} \{\tau_\infty > t\} \le 
    \mathbbm{P}_{\text{\euro{}}} \{\widehat F_t > 0\}.
  \]
  In order to show the reverse inequality in the case $\rho > 0$, we
  will repeat the previous argument but this time we use stopping times
  \[
    \tau^{(1)}_n = \inf\{\, t \in (0, T] \, ; \, F_t \ge n \,\},\quad n\in\mathbbm{N}
  \]
  and second the set of stopped foreign equivalent martingale measures
  \[
    \mathbbm{P}^{(1)}_{n,\text{\euro{}}} \{A\} = 
    \mathbbm{E}^{\mathbbm{P}_{\text{\pounds{}}}} \{F_T^{\tau^{(1)}_n} [ A ]\}.
  \]
  Repeating the same argument we notice that
  \[
    \mathbbm{E}^{\mathbbm{P}_{\text{\pounds{}}}} \{F_t\}
    =
    \lim_{n \to \infty} \mathbbm{P}^{(1)}_{n,\text{\euro{}}}\{\tau^{(1)}_n > t\}.
  \]
  By the uniqueness of the SDE we can deduce that the previous limit is
  \[
    \lim_{n \to \infty} \mathbbm{P}^{(1)}_{n,\text{\euro{}}}\{\tau^{(1)}_n > t\}
    = 
    \mathbbm{P}_{\text{\euro{}}}\{\tau^{(1)}_\infty > t\}
    = 
    \mathbbm{P}_{\text{\euro{}}}\{\widehat F_t > 0\}.
  \]
  This finally implies the equation~\eqref{thm:1:eq1} by relating the
  explosion probability with the probability of $\widehat F$ not hitting
  zero before time $t$.
\end{proof}
We should emphasize that the fact that Theorem \ref{thm:1} considers the
martingale defect only under the domestic equivalent martingale measure
poses no restriction to generality because of the well-known DOM-FOR
symmetry which we recall, for the sake of self-containedness, in the
following Lemma. 
\begin{lem}
Let $\widehat{V}$ denote plain vanilla option prices under the foreign equivalent martingale measure, then 
\begin{equation*}
 V_t(K, T, \phi)
  = \mathrm{BSM}(F_t; K, T, \phi)
  = S_t \mathrm{BSM}(\widehat{F}_t; {\mbox{$\frac{{1}}{{K}}$}}, T, -\phi)
  = S_t K \widehat{V}_t({\mbox{$\frac{{1}}{{K}}$}}, T, -\phi),
\end{equation*}
where $\phi = \pm 1$ for a call and put option, respectively.
\end{lem}
\begin{proof}
The DOM investor values the call with payoff $(S_T - K)^+$ under $\mathbbm{P}_{\text{\pounds}}$ at GBP
  \[
  V_t(K, T, 1)
  = {\mathbbm{E}}^{\mathbbm{P}_{\text{\pounds}}}\{(S_T - K)^+ \,|\, \mathcal F_t\}
  \]
 and hence EUR
  \[
  V_t(K, T, 1) / S_t.
  \]
 \noindent Since
  \[
    (S_T - K)^+ = KS_T ({\mbox{$\frac{{1}}{{K}}$}} - {\mbox{$\frac{{1}}{{S_TK}}$}})^+
  \]
  together with the change of measure formula  for the martingale
  measures for every elementary event $A\in\Omega$: $\mathbbm{P}_{\text{\euro{}}} \{A\} = \mathbbm{E}^{\mathbbm{P}_{\text{\pounds{}}}} \{F_T [ A ]\}$
  and $S_T = F_T$, we have
  \[
  \mathbbm{E}^{\mathbbm{P}_{\text{\pounds{}}}}\{(S_T - K)^+ \,|\, \mathcal F_t\}
  = K {\mathbbm{E}}^{\mathbbm{P}_{\text{\euro{}}}}\{({\mbox{$\frac{{1}}{{K}}$}} - \widehat{S}_T)^+ \,|\, \mathcal F_t\}.
  \]
\noindent This implies that we can interpret the original call in EUR as a put in GBP with payoff
  $K({\mbox{$\frac{{1}}{{K}}$}} - \widehat{S}_T)^+$. The price of this GBP put in EUR is
  \[
  K \widehat{V}_t({\mbox{$\frac{{1}}{{K}}$}}, T, -1)
  = K {\mathbbm{E}}^{\mathbbm{P}_{\text{\euro{}}}}\{({\mbox{$\frac{{1}}{{K}}$}} - \widehat{S}_T)^+ \,|\, \mathcal F_t\}.
  \]
\noindent Therefore,
  \[
  V_t(K, T, 1) / S_t
  = 
  K \widehat{V}_t({\mbox{$\frac{{1}}{{K}}$}}, T, -1)
  \]
  and the claim for $\phi = 1$ follows. The claim for $\phi = -1$ follows from this by multplying both sides by $-1$.
 \end{proof} 
\begin{remark}
From the FOR investor's perspective $\{\widehat{F}_t(T), 0\leq t\leq T\}$ has a positive mass at zero under $\mathbbm{P}_{\text{\euro{}}}$ if and only if from the DOM investor's perspective $\{F_t(T), 0\leq t\leq T\}$ is a strict local martingale under $\mathbbm{P}_{\text{\pounds}}$ and vice versa. Consider  a contingent claim paying out one EUR at $T$ and set for simplicity $r_{\text{\pounds}}=r_{\text{\euro{}}}=0$. The FOR investor will price this at one EUR at any time $t$ independent of what happens to the GBP, whereas the DOM investor values the contract at $F_0(T)^{-1}\mathbbm{E}^{\mathbbm{P}_{\text{\pounds}}}\{F_t(T)|F_0(T)\}$ which equals one EUR as long as $\{F_t(T),t\geq 0\}$ is a true martingale under $\mathbbm{P}_{\text{\pounds}}$. However, if $\{F_t(T),t\geq 0\}$ is a strict local martingale under $\mathbbm{P}_{\text{\pounds}}$, then the value obtained by the DOM investor is strictly less than one EUR at any time before $T$ which comes from the fact that under $\mathbbm{P}_{\text{\euro{}}}$, the process $\{\widehat{F}_t(T),t\geq 0\}$ can hit zero in finite time. The economic interpretation of this mathematical property is that both investors fear the extreme event of a total devaluation of the GBP which has a non-zero probability from the FOR investor's perspective and manifests itself in the presence of a strict martingale rather than a true martingale from the DOM investor's perspective. Note moreover, that the DOM investor prices the contingent claim under the condition that the GBP has some strictly positive value. When pricing under the same positivity condition, the FOR investor obtains the same value strictly less than one EUR as the DOM investor.
\end{remark}
In analogy to the risk indicator defined in the work \cite{Piiroinen} let us
define the following quantity  
\begin{equation}\label{eqn:defect}
d_{\text{\pounds}}(T;\boldsymbol{\theta})=1-F_0(T)^{-1}\mathbbm{E}^{\mathbbm{P}_{\text{\pounds}}}\{F_t(T)|F_0(T)\}=\mathbbm{P}_{\text{\euro{}}}\{\widehat{F}_t(T)=0\}
\end{equation}
for given SABR parameters $\boldsymbol{\theta}=(\alpha,\nu,\rho)^T$ which quantifies the perceived tail risk assigned to the event of a massive devaluation of the GBP against the EUR. We call this the normalized GBP \emph{SABR martingale defect} for maturity $T$.
Moreover, we define what we call the GBP \emph{SABR martingale defect indicator} via the limit
\begin{equation}\label{eqn:indicator}
A_{\text{\pounds}}(\boldsymbol{\theta}):=\lim_{T\rightarrow\infty}d_{\text{\pounds}}(T;\boldsymbol{\theta})=1-\exp(-2\rho\alpha/\nu).
\end{equation}
We prefer the indicator (\ref{eqn:indicator}) over the normalized GBP SABR martingale defect (\ref{eqn:defect}) for maturity $T$ because it enables risk comparison between different maturities. We think that this is important because when approaching Brexit related events, the market participant's expectations and fears are reflected most prominently in those options expiring shortly after the particular event so that the relevant time to maturity decreases while approaching the event. Moreover, we justify our preference by the monotony properties of the martingale defect which can be seen from Figure \ref{fig:4} below. 
\begin{figure}[h!]
\centering
\begin{picture}(300,180)
\put(0,0){\includegraphics[width=11cm]{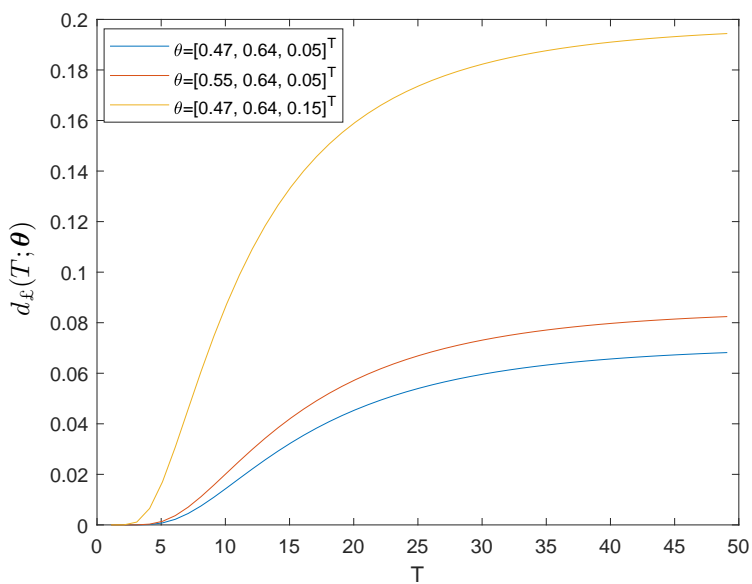}}
\put(15,105){\begin{rotate}{90}{\small$d_{\text{\pounds}}(T;\boldsymbol{\theta})$}\end{rotate}}
\end{picture}
\caption{Normalized GBP SABR martingale defect (\ref{eqn:defect}) for different parameter vectors $\boldsymbol{\theta}$ plotted against the maturity $T$ in years on the $x$-axis.}\label{fig:4}
\end{figure}
\pagebreak
\section{Statistical Method}
It has been shown in \cite{Jacquier} that in option markets where trades are
fully collateralized, the martingale defect can, in theory, be inferred
equivalently from both observed put and call implied volatility surfaces.
However, as the result is asymptotic in nature it is necessary to extrapolate
from the data observed in the market and, as in \cite{Piiroinen}, we employ the
lognormal SABR model for this task. As we have already pointed out, FX implied
volatility smiles are usually calibrated using merely $5$ broker quotes per
maturity time slice and the reliability of the $10$-delta quotes is somewhat
questionable. In order to account for this inherent calibration uncertainty, we
adopt a statistical perspective on the problem of calibrating the SABR model
(\ref{eqn:SABR1}), (\ref{eqn:SABR2}) in order to obtain the market implied SABR
martingale defect indicator (\ref{eqn:indicator}). To be precise, all
quantities are considered as random variables with certain prior distributions
that incorporate our prior knowledge about them. In this setting we can define
a statistical inverse problem whose solution is given by the posterior
distribution of the SABR parameters $\boldsymbol{\theta}= (\alpha,\nu,\rho)^T$
conditioned on the observed bid and ask market quotes which in turn yields the
posterior probability distribution of the quantity of interest, the martingale
defect indicator (\ref{eqn:indicator}), conditioned on the observed market
quotes.

\subsection{Statistical Inverse Problem}
Let $(\Omega',\mathcal{G},\mathbb{P})$ denote a probability space and let 
\begin{equation}
(\boldsymbol{\Theta},\boldsymbol{E}):\Omega'\rightarrow\mathbb{R}^{3+5},\quad \boldsymbol{Y}:\Omega'\rightarrow \mathbb{R}^5
\end{equation}
denote random vectors on this probability space. We use capital letters for random vectors and lower case letters for their realizations. The vector $(\boldsymbol{\Theta},\boldsymbol{E})$ represents the quantities that cannot be directly observed, i.e., the unknown SABR parameters $\boldsymbol{\theta}= (\alpha,\nu,\rho)^T$ and an unknown error vector $\boldsymbol{E}$ which accounts for discrepancies between SABR model and quotes observed in the market whereas $\boldsymbol{Y}$ represents the vector of mid implied volatilities. To be precise, for a fixed $T$ we define the random SABR-parameter-to-implied-volatility-map  
\begin{equation*}
(\boldsymbol{\Theta},\boldsymbol{E})\mapsto \mathcal{L}(\boldsymbol{\Theta},\boldsymbol{E})=f(\boldsymbol{\Theta})+\boldsymbol{E}=\left[\begin{array}{l}
\sigma^{\text{SABR}}(F_t;K_{10,-1},T,\boldsymbol{\Theta})\\
\sigma^{\text{SABR}}(F_t;K_{25,-1},T,\boldsymbol{\Theta})\\
\sigma^{\text{SABR}}(F_t;K_{\textrm{ATM}},T,\boldsymbol{\Theta})\\
\sigma^{\text{SABR}}(F_t;K_{25,1},T,\boldsymbol{\Theta})\\
\sigma^{\text{SABR}}(F_t;K_{10,1},T,\boldsymbol{\Theta})\end{array}\right]+\boldsymbol{E}=:\boldsymbol{Y},
\end{equation*}
where for a given realization $\boldsymbol{\theta}$, the model implied
volatility $\sigma^{\text{SABR}}(F_t;K,T,\boldsymbol{\theta})$ is computed via
the second order asymptotic formula obtained by Paulot, cf. \cite{Paulot}:
\begin{equation*}
\sigma_0(F_t;K,T,\boldsymbol{\theta})\left(1+\frac{\sigma_1(F_t;K,T,\boldsymbol{\theta})}{\sigma_0(F_t;K,T,\boldsymbol{\theta})}T+\frac{\sigma_2(F_t;K,T,\boldsymbol{\theta})}{\sigma_0(F_t;K,T,\boldsymbol{\theta})}T^2+o(T^2)\right)
\end{equation*}
with $\sigma_0(F_t;K,T,\boldsymbol{\theta}),$
$\sigma_1(F_t;K,T,\boldsymbol{\theta})$ and
$\sigma_2(F_t;K,T,\boldsymbol{\theta})$ defined as in \cite{Paulot}.
This formula can be evaluated very efficiently which is crucial with regard to the computational cost of our statistical sampling method. At the same time the formula can be shown to be highly accurate in the strike and maturity regimens which we are interested in here. The probability distribution of the random vector $\boldsymbol{Y}$ conditioned on the vectors $\boldsymbol{\theta}$ and $\boldsymbol{e}$ is given by
\begin{equation*}
 \pi(\boldsymbol{y}|\boldsymbol{\theta},\boldsymbol{e})=\delta(\boldsymbol{y}-\mathcal{L}(\boldsymbol{\theta},\boldsymbol{e})),
\end{equation*}
where $\delta$ denotes Dirac's delta in $\mathbb{R}^k$. Let $\pi_{\text{pr}}$ denote the prior probability density of $(\boldsymbol{\Theta},\boldsymbol{E})$, then we may write the joint probability density of $(\boldsymbol{\Theta},\boldsymbol{E})$ and $\boldsymbol{Y}$ as
\begin{equation}\label{eqn:den}
\pi(\boldsymbol{\theta},\boldsymbol{e},\boldsymbol{y})=\pi(\boldsymbol{y}|\boldsymbol{\theta},\boldsymbol{e})\pi_{\text{pr}}(\boldsymbol{\theta},\boldsymbol{e})=\delta(\boldsymbol{y}-\mathcal{L}(\boldsymbol{\theta},\boldsymbol{e}))\pi_{\text{pr}}(\boldsymbol{\theta},\boldsymbol{e}). 
\end{equation}
For simplicity, we assume here that $\boldsymbol{\Theta}$ and $\boldsymbol{E}$ are independent random variables. Then we obtain from (\ref{eqn:den}) by integration
\begin{equation}
\pi(\boldsymbol{\theta},\boldsymbol{y})=\pi_{\text{pr}}(\boldsymbol{\theta})\pi_{\text{noise}}(\boldsymbol{y}-\mathcal{L}(\boldsymbol{\theta})) 
\end{equation}
so that we can formulate the following statistical inverse calibration problem: 

\vspace{.25cm}
\noindent\emph{Compute the posterior distribution of $\boldsymbol{\Theta}$ conditioned on the observed market implied volatility quotes $\boldsymbol{y}$ which is given by Bayes' formula
\begin{equation}\label{eqn:Bayes}
\pi(\boldsymbol{\theta}|\boldsymbol{y})=\frac{\pi(\boldsymbol{\theta},\boldsymbol{y})}{\int_{\mathbb{R}^3} \pi(\boldsymbol{\theta},\boldsymbol{y})\mathrm{d} \boldsymbol{\theta}}.
\end{equation}
\noindent Given the (\ref{eqn:Bayes}), compute the posterior density for our quantity of interest, the martingale defect indicator 
\begin{equation}\label{eqn:stat_defect}
\pi(A_{\text{\pounds}}(\boldsymbol{\theta})|\boldsymbol{y}).
\end{equation}}

\subsection{Sampling the Posterior Density}
The unnormalized posterior density reads
\begin{equation*}
\pi(\boldsymbol \theta \vert \boldsymbol y) \propto \pi(\boldsymbol\theta) \pi(\boldsymbol y \vert \boldsymbol\theta),
\end{equation*}
where  $\pi(\boldsymbol \theta)$ and $\pi(\boldsymbol y \vert \boldsymbol\theta)$ are the prior and likelihood probability density, respectively.
We factorize the prior as 
\begin{equation*}
\pi(\boldsymbol\theta)  = \pi(\alpha)\pi(\nu)\pi(\rho) =  [\alpha\in \mathbb{R}][\nu \geq 0] [ \vert \rho\vert \leq 1 ],
\end{equation*}
where we have a flat prior for $\alpha$, flat prior in $\mathbb{R}^+$ for $\nu$, and a uniform prior for $\rho\in [-1,1]$ and for notational convenience, we have used the Iverson bracket $[\cdot ]$ as an indicator function:
\begin{equation*}
[B] := \begin{cases} 1 , & \textrm{if $B$ is true,} \\ 0, & \textrm{otherwise.} \end{cases}
\end{equation*}
Our prior construction is an improper prior, however, in practical numerical
computations in connection with the likelihood density, this posterior density
becomes a proper probability density. This means that in contrast to sampling
from an improper prior density, sampling the corresponding posterior is indeed
feasible. For a discussion on using improper priors in MCMC sampling schemes,
we refer to Hobert and Casella \cite{Hobert1996}.

We assume that the observation error is Gaussian such that the likelihood function is given by
\begin{equation*}
\pi(\boldsymbol y \vert \boldsymbol\theta) \propto \exp\left(-\frac{1}{2}  ( \boldsymbol{y}-f(\boldsymbol\theta))^T \Sigma^{-1}( \boldsymbol{y}-f(\boldsymbol\theta))\right)\prod_{i=1}^k \left[ \vert  \boldsymbol{y}_i-f_i(\boldsymbol\theta) \vert \leq \frac{1}{2}\text{BA}_i\right],
\end{equation*}
where $\Sigma$ is  covariance matrix of the observation error $\boldsymbol E$,
and $\text{BA}_i$ is the bid-ask spread at the $i$-th strike in terms of the
implied volatilities. The corresponding posterior proves to be difficult to
study analytically, as we have a non-linear parameter estimation problem with
somewhat complex priors and constraints. To cope with this complexity we use an
adaptive combined optimization and MCMC sampling algorithm. With Nelder-Mead
optimization, we compute the maximum a posteriori (map) estimate. Given this
map estimate as a start value, we use adaptive MCMC in the sense of Haario,
Saksman and Tamminen \cite{Haario} for both, obtaining the conditional mean
estimator which approximates the conditional expectation
\begin{equation*}
\mathbbm{E}^{\mathbbm{P}}\{A_{\text{\pounds}}(\boldsymbol{\Theta})|\boldsymbol{y}\}
\end{equation*}
and for providing uncertainty quantification for all the unknown parameters and the GBP martingale defect indicator by estimation of the corresponding marginal distributions.
For more details on the algorithm, we refer to the work \cite{Piiroinen}.

\section{Results}
For the observation dates listed in Appendix A, we have computed the conditional mean of the GBP SABR martingale defect indicator for the solution ({\ref{eqn:Bayes}) of the statistical inverse calibration problem conditioned on the observed market implied volatility quotes for the currency pair EURGBP. The data was obtained from Refinitiv at the New York cut at 10 a.m. ET. Brexit was originally meant to happen on March 29 2019, but on March 22 the EU-27 approved a Brexit extension until April 12. On April 10 2019, a further half-year extension was agreed between the UK and the EU-27 at the EU summit, until October 31 2019. Taking this into account, we have chosen the relevant expiry dates according to Table \ref{tab:1}. As is illustrated by Figure \ref{fig:surf}, the perceived event risk manifests itself in both magnitude and skewness of the quoted nearby volatility time slices. 
\begin{minipage}{0.94\textwidth}  \centering
\captionof{table}{Expiries used for the computation of the martingale defect indicator.}\label{tab:1}
\begin{tabularx}{.925\textwidth} {  
  | >{\centering\arraybackslash}X 
	| >{\centering\arraybackslash}X 
	| >{\centering\arraybackslash}X 
	| >{\centering\arraybackslash}X 
	| >{\centering\arraybackslash}X 
	| >{\centering\arraybackslash}X 
	| >{\centering\arraybackslash}X 
	| >{\centering\arraybackslash}X 
	| >{\centering\arraybackslash}X 
	| >{\centering\arraybackslash}X 
	| >{\centering\arraybackslash}X 
  | >{\centering\arraybackslash}X  
	| >{\centering\arraybackslash}X | }
 \hline
 15 Jan & 12 Mar & 14 Mar & 21 Mar & 29 Mar & 10 Apr & 24 May & 24 Jul & 28 Aug & 03 Sep & 09 Sep & 17 Oct & 19 Oct \\
 \hline
 2M  & 2W  & 2W & 1W & 2W & 1W & 5M & 3M & 2M & 2M & 2M & 2W & 2W \\
\hline
\end{tabularx}
\end{minipage}
\begin{figure}[t!]
\centering
\begin{picture}(320,200)
\put(0,0){\includegraphics[width=11cm]{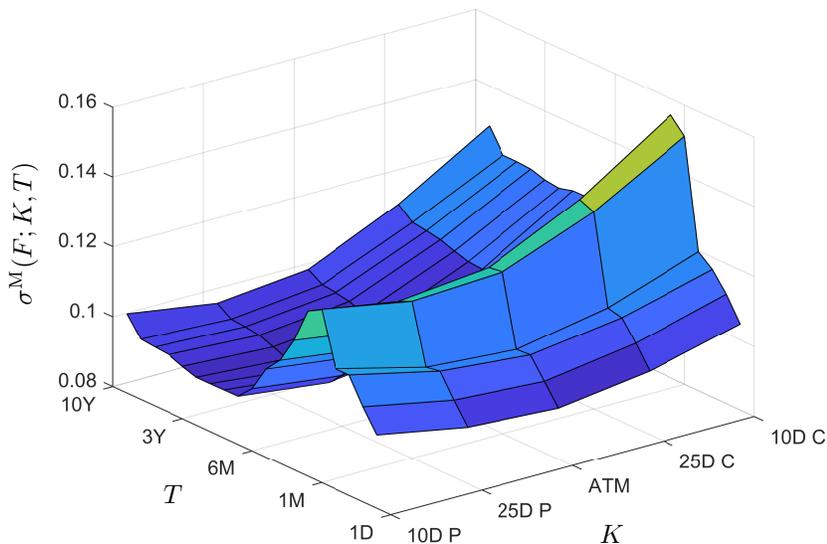}}
\put(10,105){\begin{rotate}{90}{\small$\sigma^{\text{M}}(F;K,T)$}\end{rotate}}
\put(60,30){\small$T$}
\put(225,15){\small$K$}
\end{picture}
\caption{Raw input implied volatility surface obtained from linear interpolation of the mid market quotes for August 28. The anticipated event risk with regard to October 31 is clearly visible for the expiry dates 1M and 2M.}\label{fig:surf}
\end{figure}

Our main result is the Brexit \lq\lq fever curve\rq\rq\ depicted in Figure \ref{fig:ind} which shows the conditional mean estimate obtained from sampling the posterior density (\ref{eqn:stat_defect}) of the GBP SABR martingale defect indicator (\ref{eqn:defect}) for the observation dates and expiry dates in Table \ref{tab:1}. It should be mentioned that after October 19 and until December 11 (the date of publication of the preprint of this work), the perceived tail risk measured via the martingale defect indicator stayed below $2\%$ for all expiries longer than two weeks. In particular the perceived tail risk around the general election on December 12 was mainly related to short-term moves. 
The computation is based on 100000 MCMC samples per observation day and for the observation error we have assumed zero-mean white noise with unit variance. For illustration purposes the implied volatility smile obtained from the conditional mean estimate of the random SABR parameter $\boldsymbol{\Theta}$ conditioned on the 6M observed market data $\boldsymbol{y}$ for the observation day April 10 is depicted in Figure \ref{fig:3}. Note that this expiry is different from the one we have used for computing the respective point in the \lq\lq fever curve\rq\rq\ which is 1W. The background is that on April 10 a half-year extension was agreed between the UK and the EU-27, until October 31 2019. As a result, during that trading day the pronounced skewness of the implied volatility smile shifted from the short term expiry dates to those expiry dates around October 31, i.e., 5M, 6M and 9M. This is visible in the corresponding GBP SABR martingale defect indicators. While the perceived 1W short term risk measured by our MCMC approximation of $\mathbbm{E}^{\mathbbm{P}}\{A_{\text{\pounds}}(\boldsymbol{\Theta})|\boldsymbol{y}\}$ is merely $1.40\%$, the perceived 6M risk corresponding to the smile shown in Figure \ref{fig:3} is comparably high, namely $8.75\%$. That is, within a couple of days, the time horizon of the perceived risk has shifted as a result of the political events that finally let to the extension. 

Moreover, it can be see from the plot that the SABR model fits the market smile well inside the bid-ask spread. The fact that these bid-ask spreads are rather wide for the out-of-the-money-options is linked to the issue of non-uniqueness for the solution of the classical deterministic calibration problem and underlines the usefulness of our statistical approach enabling parameter uncertainty quantification. This uncertainty quantification is illustrated in Figure \ref{fig:chain}, where we show traceplots and plots of cumulative averages for the estimates of $\alpha$, $\rho$, $\nu$ and the martingale defect indicator $A_{\text{\pounds}}(\boldsymbol{\theta})$. We also plot the marginal densities obtained by removing the burn-in period (25\% of the whole chain length), and by using a kernel density estimator with Epanechnikov kernels. By visually assessing, we note that we have good mixing of the chains, however, without the optimization step for choosing the start value, the chains would not converge.

\begin{figure}[t!]
\centering
\begin{picture}(320,200)
\put(0,0){\includegraphics[width=11cm]{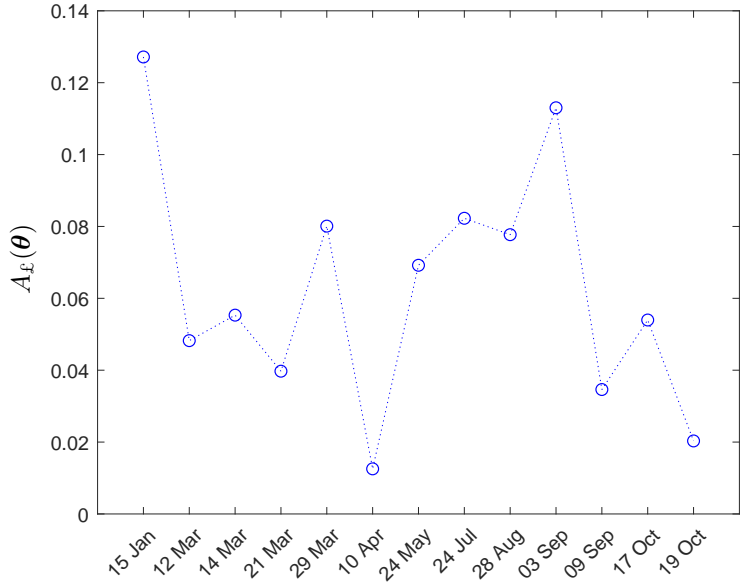}}
\put(15,110){\begin{rotate}{90}{\small$A_{\text{\pounds}}(\boldsymbol{\theta})$}\end{rotate}}
\end{picture}
\caption{The Brexit \lq\lq fever curve\rq\rq\ given by the conditional mean estimates obtained from sampling the posterior density (\ref{eqn:stat_defect}).}\label{fig:ind}
\end{figure}

\begin{figure}[t!]
\centering
\begin{picture}(320,200)
\put(0,0){\includegraphics[width=11cm]{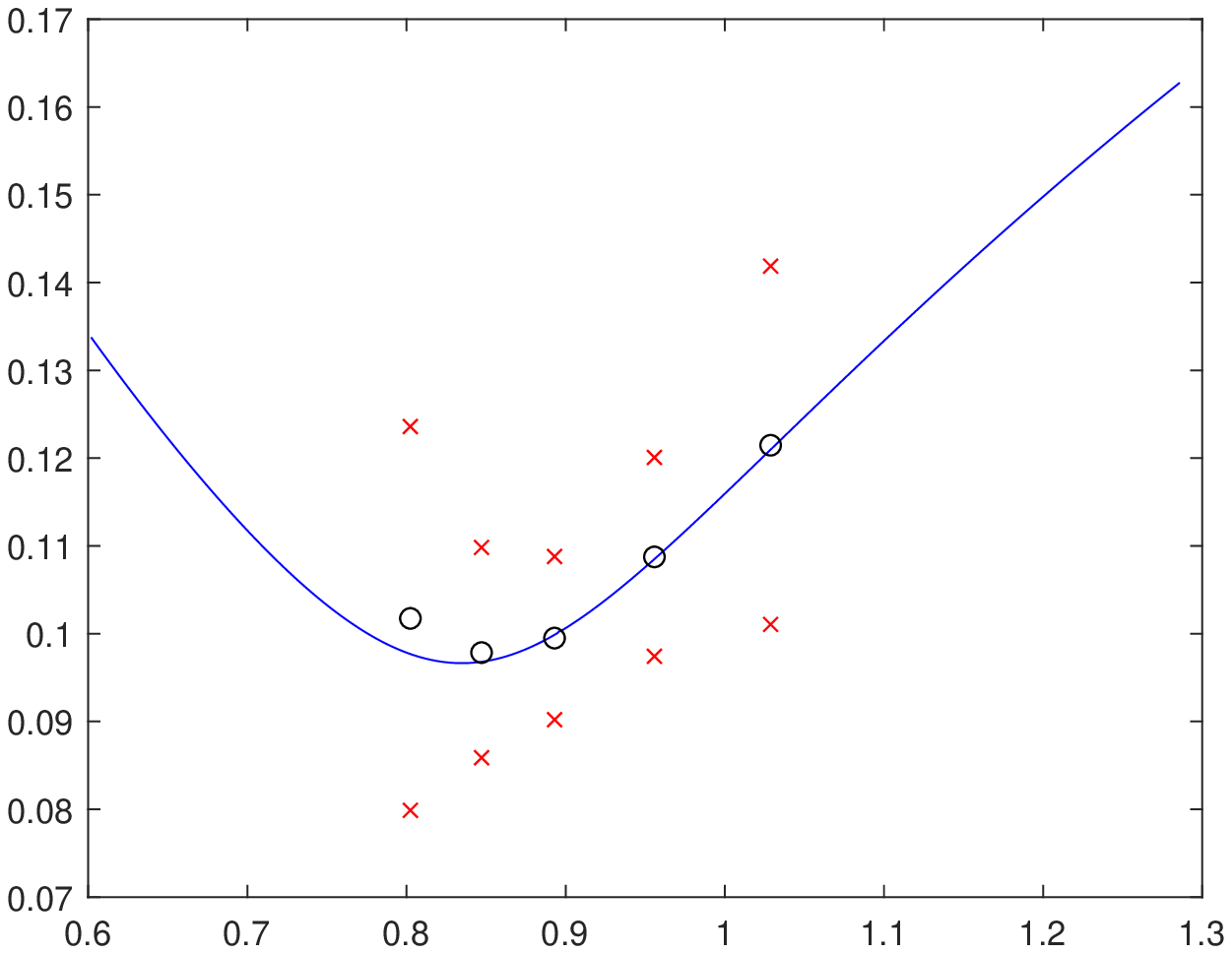}}
\put(15,84){\begin{rotate}{90}{\small$\sigma^{\text{SABR}}(F_t;K,T,\boldsymbol{\theta})$}\end{rotate}}
\put(160,3){{\small$K$}}
\end{picture}
\caption{April 10 implied volatility smile obtained from the conditional mean of the random SABR parameter vector $\boldsymbol{\Theta}$ conditioned on the 6M observed bid and ask volatility quotes given by the red x markers (o markers are the corresponding mids). The corresponding GBP SABR martingale defect indicator is $8.75\%$.}\label{fig:3}
\end{figure}

\begin{figure}[b!]
\centering
\begin{picture}(320,260)
\put(0,0){\includegraphics[width=11cm]{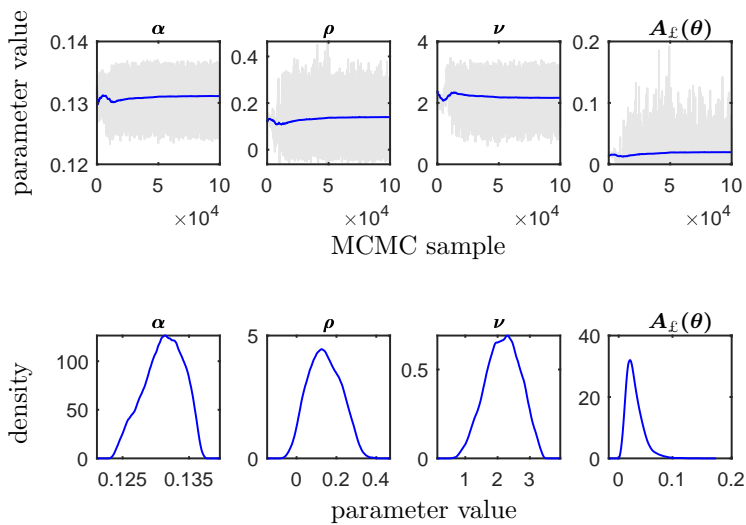}}
\put(250,203){{\footnotesize$\boldsymbol{A_{\text{\pounds}}(\boldsymbol{\theta})}$}}
\put(250,92){{\footnotesize$\boldsymbol{A_{\text{\pounds}}(\boldsymbol{\theta})}$}}
\put(130,120){{\small MCMC sample}}
\put(15,142){\begin{rotate}{90}{\small parameter value}\end{rotate}}
\put(15,48){\begin{rotate}{90}{\small density}\end{rotate}}

\put(130,20){{\small parameter value}}
\end{picture}
\caption{October 19 sampled MCMC chains (upper row) and corresponding estimated densities (lower row).}\label{fig:chain}
\end{figure}

\clearpage
\newpage
\mbox{~}

\section{Conclusion}
We found that the martingale defect which occurs in the lognormal SABR model in the presence of extremely skewed implied volatility smiles may be used to asses the market expectations of extreme events such as a no-deal Brexit. For this purpose we introduced a statistical SABR martingale defect indicator which quantifies the market expectation for the GBP to progressively depreciate against the EUR based on observed EURGBP option prices. This forward-looking measure of market expectations accounts for the inherent uncertainty due to the small number of reliable volatility quotes observable in the market and it should be of great use for risk management purposes such as data-driven risk scenario generation or stress testing. Finally, we would like to point out that the \lq\lq fever curve\rq\rq\ we have computed for a timeline of Brexit related events in 2019 quantifies remarkably well the public perception of economic risk related to a no-deal Brexit scenario.

\section{Afterword}
This paper is the final version of a preprint dated December 11 2019 that originally appeared on arXiv (https://arxiv.org/abs/1912.05773) on December 12 2019, in advance of the general election on the same day.


\section*{Funding}
This work has been funded by Academy of Finland (decision numbers 326240 and 326341, and Finnish Centre of
Excellence in Inverse Modelling and Imaging, decision number 312119).

\bibliographystyle{rQUF}

\pagebreak

\appendix
\newpage
\section{Timeline of events for the tested dates}
\nopagebreak
\begin{table}[h]
\caption{2019 Brexit Timeline part 1/2. Source: Wikipedia}
\centering
\begin{minipage}[t]{1.\linewidth}
\color{gray}
\rule{\linewidth}{1pt}
\ytl{15 Jan}{The First meaningful vote is held on the Withdrawal Agreement in the UK House of Commons. The UK Government is defeated by 432 votes to 202}
\ytl{12 Mar}{The Second meaningful vote on the Withdrawal Agreement with the UK Government is defeated again by 391 votes to 242}
\ytl{14 Mar}{The UK Government motion passes 412 to 202 to extend the Article 50 period}
\ytl{21 Mar}{The European Council offers to extend the Article 50 period until 22 May 2019 if the Withdrawal Agreement is passed by 29 March 2019 but, if it does not, then the UK has until 12 April 2019 to indicate a way forward. The extension is formally agreed the following day}
\ytl{29 Mar}{The original end of the Article 50 period and the original planned date for Brexit. Third vote on the Withdrawal Agreement after being separated from the Political Declaration. UK Government defeated again by 344 votes to 286}
\ytl{10 Apr}{The European Council grants another extension to the Article 50 period to 31 October 2019, or the first day of the month after that in which the Withdrawal Agreement is passed, whichever comes first}
\ytl{24 May}{Theresa May announces that she will resign as Conservative Party leader, effective 7 June, due to being unable to get her Brexit plans through parliament and several votes of no-confidence, continuing as prime minister while a Conservative leadership contest takes place}
\ytl{24 Jul}{Boris Johnson accepts the Queen's invitation to form a government and becomes Prime Minister of the United Kingdom, the third since the referendum}

\bigskip
\rule{\linewidth}{1pt}%
\end{minipage}%
\end{table}

\begin{table}
\caption{2019 Brexit Timeline part 2/2. Source: Wikipedia}
\centering
\begin{minipage}[t]{1.\linewidth}
\color{gray}
\rule{\linewidth}{1pt}
\ytl{28 Aug}{Boris Johnson announces his intention to prorogue Parliament in September}
\ytl{3 Sep}{A motion for an emergency debate to pass a bill that would rule out a unilateral no-deal Brexit by forcing the Government to get parliamentary approval for either a withdrawal agreement or a no-deal Brexit. This motion, to allow the debate for the following day, passed by 328 to 301. 21 Conservative MPs voted for the motion}
\ytl{9 Sep}{The Government again loses an attempt to call an election under the Fixed-term Parliaments Act. Dominic Grieve's humble address, requiring key Cabinet Office figures to publicise private messages about the prorogation of parliament, is passed by the House of Commons. Speaker John Bercow announces his intention to resign as Speaker of the House of Commons on or before 31 October. The Benn Bill receives Royal Assent and becomes the European Union (Withdrawal) (No. 2) Act 2019. Parliament is prorogued until 14 October 2019. Party conference season begins, with anticipation building around a general election}
\ytl{17 Oct}{The UK and European Commission agree on a revised withdrawal agreement containing a new protocol on Northern Ireland. The European Council endorses the deal}
\ytl{19 Oct}{A special Saturday sitting of Parliament is held to debate the revised withdrawal agreement. The prime minister moves approval of that agreement. MPs first pass, by 322 to 306, Sir Oliver Letwin's amendment to the motion, delaying consideration of the agreement until the legislation to implement it has been passed; the motion is then carried as amended, implementing Letwin's delay. This delay activates the Benn Act, requiring the prime minister immediately to write to the European Council with a request for an extension of withdrawal until 31 January 2020}
\bigskip
\rule{\linewidth}{1pt}%
\end{minipage}%
\end{table}

\end{document}